\documentclass[11pt]{article}

\usepackage[utf8]{inputenc}

\usepackage[letterpaper,margin=1in]{geometry}

\usepackage{amsmath,amssymb}
\usepackage{amsthm}

\usepackage{complexity}

\usepackage{dsfont}
\usepackage{bbm}

\usepackage{hyperref}
\hypersetup{colorlinks, linkcolor=darkblue, citecolor=darkgreen, urlcolor=darkblue}
\usepackage[capitalize, nameinlink]{cleveref}
\crefname{theorem}{Theorem}{Theorems}
\Crefname{lemma}{Lemma}{Lemmas}
\Crefname{claim}{Claim}{Claims}
\Crefname{fact}{Fact}{Facts}
\Crefname{remark}{Remark}{Remarks}
\Crefname{observation}{Observation}{Observations}
\Crefname{figure}{Figure}{Figures}
\Crefname{line}{Line}{Lines}
\Crefname{algocf}{Algorithm}{Algorithms}
\crefalias{AlgoLine}{line}
\Crefname{stepsromani}{Step}{Steps}
\Crefname{stepsarabici}{Step}{Steps}

\newtheorem{theorem}{Theorem}
\newtheorem{lemma}[theorem]{Lemma}

\newtheorem{proposition}[theorem]{Proposition}
\newtheorem{definition}[theorem]{Definition}
\newtheorem{remark}[theorem]{Remark}

\usepackage{times}

\usepackage[T1]{fontenc}

\usepackage{fix-cm}

\usepackage[inline]{enumitem}
\setlist[enumerate,1]{label=(\roman*), leftmargin=2.2em}
\setlist[enumerate,2]{label=(\alph*)}
\setlist{nosep,topsep=0.1em}
\setlist[itemize,1]{label={\bfseries--}}

\newlist{stepsroman}{enumerate}{1}
\setlist[stepsroman]{label=(\roman*), leftmargin=2.2em}

\newlist{stepsarabic}{enumerate}{1}
\setlist[stepsarabic]{rightmargin=0.2em, label=\arabic*., ref=\arabic*}

\usepackage{mathtools}

\usepackage[linewidth=0.8pt]{mdframed}
\mdfsetup{skipabove=\bigskipamount, skipbelow=\bigskipamount, innerleftmargin=0.5em, innertopmargin=0.5em, innerrightmargin=0.5em, innerbottommargin=0.5em, align=center}

\usepackage[color=darkgreen!40!white]{todonotes}
\setuptodonotes{size=\scriptsize}

\usepackage{xcolor}
\definecolor{darkblue}{rgb}{0,0,0.38}
\definecolor{darkred}{rgb}{0.6,0,0}
\definecolor{darkgreen}{rgb}{0.1,0.35,0}

\usepackage[
backend=biber,
style=alphabetic,
citestyle=alphabetic,
maxalphanames=6,
maxcitenames=99,
mincitenames=98,
maxbibnames=99,
giveninits=true,
url=false,
doi=true,
isbn=true,
backref=true
]{biblatex}

\usepackage{xpatch}

\makeatletter
\patchcmd\blx@bblinput{\blx@blxinit}
                      {\blx@blxinit
                      }{}{\fail}
\makeatother

\usepackage{xspace}

\makeatletter
 \newcommand{\linkdest}[1]{\Hy@raisedlink{\hypertarget{#1}{}}}
\makeatother
\newcommand{\PCTSP}{\protect\hyperlink{prb:PCTSP}{PCTSP}\xspace}
\newcommand{\PCS}{\protect\hyperlink{prb:PCS}{PCS}\xspace}
\newcommand{\TSP}{\protect\hyperlink{prb:TSP}{TSP}\xspace}
\newcommand{\pathTSP}{\protect\hyperlink{prb:pathTSP}{Path TSP}\xspace}

\newcommand{\Exp}{\mathbb{E}}
\newcommand{\Prob}{\mathbb{P}}
\newcommand{\odd}{\operatorname{odd}}
\newcommand{\core}{\operatorname{core}}
\newcommand{\drop}{{\delta}}
\newcommand{\dd}{\operatorname{d}\!}

\newcommand{\PQJ}[1][Q]{P_{#1\text{-join}}^{\uparrow}}

\renewcommand{\PP}[2]{{\mathbb{P}}_{#1}\left[#2\right]}

\renewcommand{\EE}[2]{{\mathbb{E}}_{#1}\left[#2\right]}

\newcommand{\spare}[1]{\text{spare}_{#1}}

\makeatletter
\DeclareFieldFormat{eprint:arxiv}{%
  arXiv\addcolon\space
  \ifhyperref
    {\href{https://arxiv.org/abs/#1}{%
       \nolinkurl{#1}%
       \iffieldundef{eprintclass}
     {}
     {\addspace\mkbibbrackets{\thefield{eprintclass}}}}}
    {\nolinkurl{#1}
     \iffieldundef{eprintclass}
       {}
       {\addspace\mkbibbrackets{\thefield{eprintclass}}}}}
\makeatother

\makeatletter
\newcommand{\labeltarget}[1]{\Hy@raisedlink{\hypertarget{#1}{}}}
\makeatother

\usepackage{wrapfig}

\usepackage[margin=25pt,font=small,labelfont=bf]{caption}

\usepackage{subcaption}

\usepackage{graphicx}

\graphicspath{{.}{graphics/}}
\makeatletter
\newcommand\appendtographicspath[1]{%
  \g@addto@macro\Ginput@path{#1}%
}
\makeatother

\usepackage{tikz}
\usetikzlibrary{fit, shapes, calc}

\usepackage[vlined,ruled,algo2e]{algorithm2e}
\SetAlCapHSkip{0.2em}
\SetAlgoInsideSkip{smallskip}
\AtBeginDocument{%
  \SetCustomAlgoRuledWidth{0.95\textwidth}%
  \setlength{\algowidth}{0.975\textwidth}%
}
\makeatletter
\patchcmd{\@algocf@start}
  {\begin{lrbox}{\algocf@algobox}}
  {%
   \hspace*{0.025\textwidth}%
   \begin{lrbox}{\algocf@algobox}%
   \begin{minipage}{0.95\textwidth}%
  }
  {}{}
\patchcmd{\@algocf@finish}
  {\end{lrbox}}
  {\end{minipage}\end{lrbox}}
  {}{}
\makeatother

\usepackage{xfrac}
\ExplSyntaxOn
\DeclareRestrictedTemplate { xfrac } { text } { math }
  {
    numerator-font      = \number \fam ,
    slash-symbol        = /            ,
    slash-symbol-font   = \number \fam ,
    denominator-font    = \number \fam ,
    scale-factor        = 0.7          ,
    scale-relative      = false        ,
    scaling             = true         ,
    denominator-bot-sep = 0 pt         ,
    math-mode           = true         ,
    phantom             = ( %
  }
\DeclareInstance { xfrac } { mathdefault } { math }
  { numerator-top-sep = 0pt }
\ExplSyntaxOff

\let\oldtop\top
\renewcommand{\top}{{\scriptscriptstyle{\oldtop}}}
\newcommand{\transpose}{{\scriptscriptstyle{\oldtop}}}

\makeatletter
\def\@fnsymbol#1{\ensuremath{\ifcase#1\or *\or %
\ddagger\or
    \mathsection\or \mathparagraph\or \|\or **\or \dagger\dagger
    \or \ddagger\ddagger \else\@ctrerr\fi}}
\makeatother

\title{A Better-Than-1.6-Approximation for Prize-Collecting TSP}

\author{%
Jannis Blauth%
\footnote{%
Research Institute for Discrete Mathematics and Hausdorff Center for Mathematics, University of Bonn, Bonn, Germany.
Email:
\href{mailto:blauth@or.uni-bonn.de}{blauth@or.uni-bonn.de}.
}%
\and
Nathan Klein%
\footnote{%
Boston University, Boston, MA, USA.
Email:
\href{mailto:nklei1@bu.edu}{nklei1@bu.edu}.
Research supported by Air Force Office of Scientific Research grant FA9550-20-1-0212 and NSF grants DGE-1762114 and CCF-1813135.
}%
\and
Martin Nägele%
\footnote{%
Department of Mathematics, ETH Zurich, Zurich, Switzerland.
Email:
\href{mailto:martinn@ethz.ch}{martinn@ethz.ch}. %
Partially supported by the Swiss National Science Foundation (grant no.\ P500PT\_206742) and the Deutsche Forschungsgemeinschaft (DFG, German Research Foundation) under Germany's Excellence Strategy~--~EXZ-2047/1~--~390685813.
Most of this work was done while M.~Nägele was employed at University of Bonn.
}%
}%
\date{}

\begin{document}

\maketitle

\begin{abstract}
Prize-Collecting TSP is a variant of the traveling salesperson problem where one may drop vertices from the tour at the cost of vertex-dependent penalties.
The quality of a solution is then measured by adding the length of the tour and the sum of all penalties of vertices that are not visited.
We present a polynomial-time approximation algorithm with an approximation guarantee slightly below $1.6$, where the guarantee is with respect to the natural linear programming relaxation of the problem.
This improves upon the previous best-known approximation ratio of $1.774$.
Our approach is based on a known decomposition for solutions of this linear relaxation into rooted trees.
Our algorithm takes a tree from this decomposition and then performs a pruning step before doing parity correction on the remainder.
Using a simple analysis, we bound the approximation guarantee of the proposed algorithm by $\sfrac{(1+\sqrt{5})}{2} \approx 1.618$, the golden ratio.
With some additional technical care we further improve the approximation guarantee to $1.599$.
Furthermore, we show that for the path version of Prize-Collecting TSP (known as Prize-Collecting Stroll) our approach yields an approximation guarantee of  $1.6662$, improving upon the previous best-known guarantee of $1.926$.

\end{abstract}

\thispagestyle{empty}
\addtocounter{page}{-1}

\begin{tikzpicture}[overlay, remember picture, shift = {(current page.south east)}]
\node[anchor=south east, outer sep=5mm] {
\begin{tikzpicture}[outer sep=0] %
\node (dfg) {\includegraphics[height=12mm]{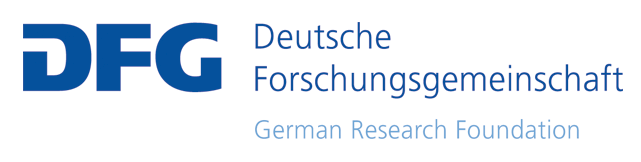}};
\node[left=5mm of dfg, yshift=1mm] (snf) {\includegraphics[height=8mm]{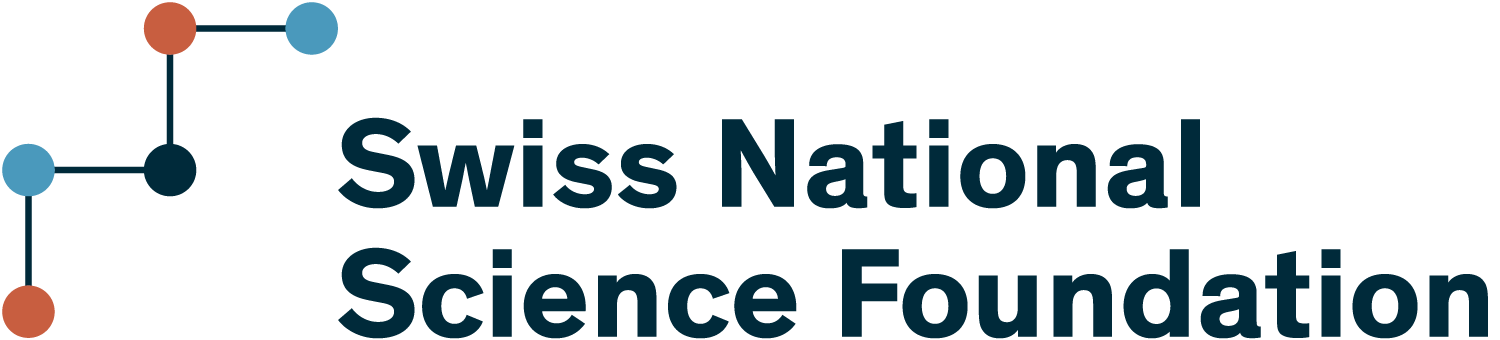}};
\node[left=5mm of snf] (nsf) {\includegraphics[height=13mm]{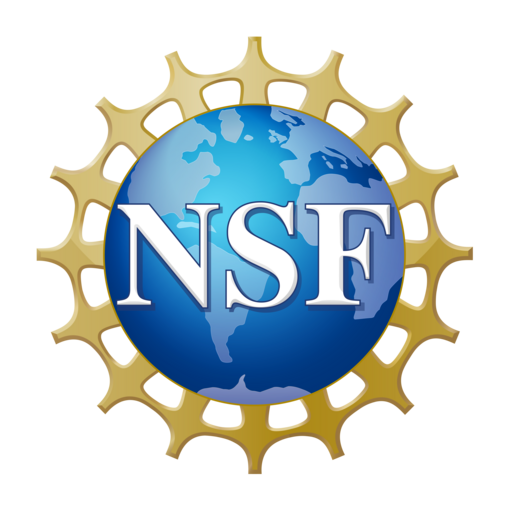}};
\end{tikzpicture}
};
\end{tikzpicture}

\newpage

\section{Introduction}

The metric traveling salesperson problem (\linkdest{prb:TSP}{\TSP}) is one of the most fundamental problems in combinatorial optimization.
In an instance of this problem, we are given a set $V$ of $n$ vertices along with their pairwise symmetric distances, $c\colon V\times V \to\mathbb{R}_{\geq 0}$, which form a metric.
The goal is to find a shortest possible Hamiltonian cycle.
In the classical interpretation, there is a salesperson who needs to visit a set of cities $V$ and wants to minimize the length of their tour.
In this work, we study a variant known as \textit{prize-collecting} \TSP, in which the salesperson can decide whether or not to include each city besides the starting one%
\footnote{``Rooted'' and unrooted versions are reducible to one another while preserving approximability, as noted in, e.g.,~\cite{archer_2011_improved}.
Here, we always require a root vertex $r$.}
in their tour at the cost of a city-dependent penalty.
Formally, the problem can be stated as follows.

\begin{mdframed}[userdefinedwidth=0.95\linewidth]
\linkdest{prb:PCTSP}{\textbf{Prize-Collecting TSP (\PCTSP):}}
Given a complete undirected graph $G=(V,E)$ with metric edge lengths $c_e\geq 0$ for all $e\in E$, a root $r \in V$, and penalties $\pi_v\geq 0$ for all $v\in V \setminus \{r\}$, the task is to find a cycle $C=(V_C, E_C)$ in $G$ that contains the root $r$ and minimizes
$$\sum_{e\in E_C} c_{e} + \sum_{v\in V\setminus V_C} \pi_v\enspace.$$
\end{mdframed}

\noindent
This is a very natural generalization of \TSP (one can recover \TSP by setting $\pi_v = \infty$ for all $v\in V\setminus\{r\}$), as from the salesperson's perspective some cities may not be worth visiting if they significantly increase the length of the tour.
Indeed, in many real-world settings instances of TSP are actually prize-collecting.

As mentioned, \PCTSP is at least as hard as \TSP.
Thus, it is \NP-hard to approximate within a factor of $\sfrac{123}{122}$ \cite{karpinski_2015_inapproximability}.
On the positive side, the first constant-factor approximation algorithm for \PCTSP was shown in the early~'90s \cite{bienstock_1993_note}, giving a ratio of~$2.5$.
After a series of improvements \cite{goemans_1995_general,archer_2011_improved,goemans_2009_combining}, the best approximation factor is now slightly below~$1.774$  \cite{blauth_2023_improved}.

In \TSP and many of its variants, such approximation guarantees typically rely on lower bounds obtained through \emph{linear programming relaxations}.
For \PCTSP, the natural such formulation is the following:\footnote{We use $\pi_r\coloneqq0$ for convenience.
For $S\subseteq V$, we denote $\delta(S)\coloneqq \{e\in E\colon |e\cap S|=1\}$; for $v\in V$, we use $\delta(v)\coloneqq \delta(\{v\})$.}
\begin{equation}\tag{PCTSP LP relaxation}\label{eq:rel_PCTSP}
\begin{array}{rrcll}
\min & \displaystyle \sum_{e\in E}c_ex_e + \sum_{v\in V}\pi_v (1-y_v) \\
     & x(\delta(v)) & = & 2y_v & \forall v\in V\setminus\{r\} \\
     & x(\delta(r)) & \le & 2 \\
     & x(\delta(S)) & \geq & 2y_v & \forall S\subseteq V\setminus\{r\}, v\in S\\
     & y_r & = & 1 \\
     & x_e & \geq & 0 & \forall e\in E\\
     & y_v & \geq & 0 & \forall v\in V\enspace.
\end{array}
\end{equation}
One can see that $y_v\leq 1$ is implied by the above formulation, hence the variables $y_v$ can be interpreted as the extent to which the vertex $v$ is visited by the fractional solution.
In this paper, we prove the following.

\begin{theorem}\label{thm:main}
There is a polynomial-time LP-relative $1.599$-approximation algorithm for \PCTSP.
\end{theorem}

To obtain this result, we exploit a known decomposition of solutions $(x,y)$ to the above relaxation into trees (see \cref{lem:tree_decomposition} for the formal statement), which can be obtained very similar to an existential result on packing branchings in a directed multigraph by \textcite[Theorem~2.6]{bang-jensen_1995_preserving} or a polynomial-time version by \textcite[Theorem~3.1]{post_2015_lpbased} and was---in a generalized form---first used in the context of \PCTSP by \textcite{blauth_2023_improved}.
The decomposition can be interpreted as a distribution $\mu$ over a polynomial number of trees $\mathcal{T}$ rooted at $r$ such that for each tree $T \in \mathcal{T}$
\begin{enumerate*}
\item $\EE{T \sim \mu}{c(E[T])} \le c^\transpose x$ and
\item $\PP{T \sim \mu}{v \in V[T]} = y_v$ for all $v \in V$.
\end{enumerate*}%
\footnote{For a (sub-)graph $H = (V_H,E_H)$, we write $V[H] \coloneqq V_H$ and $E[H] \coloneqq E_H$. Moreover, for a set $F$ of edges we abbreviate $c(F) \coloneqq \sum_{e \in F} c_e$.}

Our algorithm proceeds as follows.
We apply the decomposition to a slightly modified LP solution with $y_v = 0$ or $y_v \ge \delta$ for each $v \in V$ for some parameter~$\delta$.
Then, for a tree $T$ in the support of $\mu$ and a threshold $\gamma$, we prune the tree.
Concretely, we find the inclusion-wise minimal subtree of $T$ which spans all vertices $v \in V[T]$ with $y_v \ge \gamma$.
Finally, we add the minimum cost matching on the odd degree vertices of this subtree.
While our algorithm simply tries all possible trees $T$ in the support of $\mu$ and all possible thresholds $\delta, \gamma \in \{y_v : v \in V\}$, our analysis is randomized:
We sample the tree from $\mu$ and the thresholds $\delta$ and $\gamma$ from well-chosen distributions, and prove the main result in expectation.
Clearly, the same guarantee then holds for the best choice of~$T$, $\delta$ and $\gamma$.

\bigskip

The arguably most prominent generalization of \TSP is its path version (\linkdest{prb:pathTSP}{\pathTSP}), where in addition to the vertex set $V$ and the distances $c$, two vertices $s,t\in V$ are given, and the goal is to find a shortest path with endpoints~$s$ and~$t$ that covers all vertices in $V$.
In this work, we also study a path version of \PCTSP.
Following previous literature (see, e.g., \cite{k_stroll}), we use the term \emph{Prize-Collecting Stroll} for the resulting problem, which is formally defined as follows.%
\footnote{%
Here, an \emph{$s$-$t$ stroll} is a path from $s$ to $t$ in the underlying graph $G$.
The term \emph{stroll} is primarily used to emphasize that the path is not required to be Hamiltonian, i.e., it does not necessarily cover all vertices of $G$.
}

\begin{mdframed}[userdefinedwidth=0.95\linewidth]
\linkdest{prb:PCS}{\textbf{Prize-Collecting Stroll (\PCS):}}
Given a complete undirected graph $G=(V,E)$ with metric edge lengths $c_e\geq 0$ for all $e\in E$, two vertices~$s,t\in V$, and penalties $\pi_v\geq 0$ for all $v\in V \setminus \{s,t\}$, the task is to find an $s$-$t$ stroll $S=(V_S, E_S)$ in $G$ that minimizes
$$\sum_{e\in E_S} c_{e} + \sum_{v\in V\setminus V_S} \pi_v\enspace.$$
\end{mdframed}

\noindent Generally, techniques exploited for \PCTSP can often also be applied to \PCS with some adaptions, and typically slight losses in terms of the approximation guarantee.
The best explicit such adaption in literature is due to \textcite{an_path_tsp}. Using the best-known LP-relative approximation algorithm for \pathTSP \cite{traub_path_tsp, zhong_path_tsp}, they obtain an approximation guarantee of $1.926$.%
\footnote{%
Techniques of \cite{blauth_2023_improved} could be exploited to improve over the approximation factor of \cite{an_path_tsp}, but would---to the best of our knowledge---not beat the results presented here.%
}
We significantly improve over the state of the art by showing that our techniques allow for obtaining the following.

\begin{theorem}\label{thm:PCS}
There is a polynomial-time LP-relative $1.6662$-approximation algorithm for \PCS.
\end{theorem}

More concretely, we show that a straightforward extension of our \PCTSP algorithm readily results in a~$\sfrac53$-approximation algorithm.
We then observe that this algorithm incurs unbalanced losses on the edge cost and penalty terms: While we lose a factor of $\sfrac53$ on the edge costs compared to a lower bound obtained through an LP relaxation, only a factor of $\sfrac32$ is lost on the penalty side.
An adaption of the classical threshold rounding approach of \textcite{bienstock_1993_note} can be tuned to give an imbalance in the other direction. So, the better of the two algorithms can achieve an approximation guarantee just slightly below $\sfrac53$.
We remark that all these guarantees are with respect to the natural LP relaxation for \PCS.

\bigskip

As mentioned, our improved approximation ratio of $1.599$ for \PCTSP improves upon the previous $1.774$-approximation by \textcite{blauth_2023_improved}.
It remains open whether there is an efficient algorithm for \PCTSP{} that matches---in terms of the approximation factor---the $\sfrac{3}{2}$-approximation for \TSP by \textcite{christofides_1976_worst} and \textcite{serdyukov_1978_onekotorykh} (also see \cite{christofides_1976_worst_new,vanBevern_2020_historical}), or the current best known approximation guarantee for \TSP, which is just slightly below $\sfrac{3}{2}$ \cite{karlin_2021_slightly,karlin_2023_derand}.
The ideal result for \PCTSP would be an algorithm that, given an $\alpha$-approximation for \TSP, produces an $\alpha$-approximation for \PCTSP (or possibly an $(\alpha+\varepsilon)$-approximation for every $\varepsilon>0$).
Such a result was recently shown for \pathTSP \cite{traub_2020_reducing}, and as approximation algorithms for \PCTSP begin to approach the threshold $\sfrac{3}{2}$, this possibility feels less out of reach.

In this respect, the situation is slightly different for \PCS.
While the approximation ratio for \pathTSP is currently about $\sfrac{3}{2}$ (as discussed), all known algorithms for \PCS build upon \emph{LP-relative} approximation algorithms for \pathTSP, and the best known LP-relative algorithm is still roughly~$1.528$ \cite{traub_path_tsp, zhong_path_tsp}.
It would be of significant interest to match the latter guarantee for \PCS, or see how non-LP-relative techniques can be applied in a prize-collecting framework. Both of these goals appear to be beyond current knowledge.

\subsection[Prior work on PCTSP and PCS]{Prior work on \PCTSP and \PCS}

While \textcite{balas_1989_prize} was the first to study prize-collecting variations of \TSP, the first constant-factor approximation algorithm for \PCTSP was given by \textcite{bienstock_1993_note} through a simple threshold rounding approach:
Starting from a solution $(x,y)$ of the \ref{eq:rel_PCTSP}, the Christofides-Serdyukov algorithm is used to construct a tour on all vertices $v \in V$ with $y_v \ge \sfrac35$, giving an LP-relative $\sfrac52$-approximation.
\textcite{goemans_1995_general} later obtained a $2$-approximation through a classical primal-dual approach.
More precisely, they showed how to compute a tree $T$ with $c(E[T]) \le c^{\transpose} x$ and $\pi(V \setminus V[T]) \le \pi^{\transpose} (1-y)$, so that doubling the tree yields the $2$-approximation.%
\footnote{%
\label{fn:simple_2_approx}%
As observed in \cite{blauth_2023_improved}, such a tree---and thus an algorithm matching the guarantee of the primal-dual approach---can immediately be obtained from the decomposition of solutions of the \ref{eq:rel_PCTSP} that was mentioned earlier and that is also used in this paper.
Therefore, there are two elementary ways to get a $2$-approximation.
}

The factor of $2$ was first beaten by \textcite{archer_2011_improved}.
As a black-box subroutine, they use an approximation algorithm for \TSP which we assume has ratio $\rho$.
They achieved a $2-\big(\frac{2-\rho}{2+\rho}\big)^2$ approximation, which---for $\rho = \sfrac32$---is approximately $1.979$.
Their algorithm runs the primal-dual algorithm of \citeauthor{goemans_1995_general} and the $\rho$-approximation algorithm for \TSP on a carefully selected node set, and outputs the better of the two tours.
\textcite{goemans_2009_combining} then observed that running both threshold rounding for different thresholds and the primal-dual algorithm and choosing the best among the computed solutions yields an approximation guarantee of $\sfrac{1}{(1-\frac{1}{\beta} e^{1-\sfrac{2}{\beta}})}$, where $\beta$ denotes the approximation guarantee of an LP-relative approximation algorithm for \TSP that is used in a black-box way.
For $\beta=\sfrac32$, the guarantee of \citeauthor{goemans_2009_combining} equals approximately $1.914$.
\citeauthor{goemans_2009_combining} was the first to exploit a randomized analysis of threshold rounding, in which the threshold $\gamma$ is chosen from a specific distribution.

\textcite{blauth_2023_improved} refined the threshold rounding approach by sampling a connected subgraph such that each vertex $v \in V$ with $y_v \ge \gamma$ (again,~$\gamma$ denotes the threshold) is always contained in the vertex set of this subgraph, whereas each vertex $v\in V$ with $y_v<\gamma$ is contained with probability at least~$\exp\left(\sfrac{{-}3 y_v}{4 \gamma}\right)$.
Since each vertex below the threshold is guaranteed to have even degree in this subgraph, parities can be corrected at no extra cost, yielding an approximation guarantee of slightly below $1.774$ through a randomized analysis.
This guarantee beats those of the previously mentioned algorithms even for $\rho=1$ and $\beta=\sfrac43$ (the integrality gap of the linear programming relaxation for \TSP{} used in \cite{goemans_2009_combining}, the Held-Karp relaxation, is at least $\sfrac43$).
Although the high-level idea of \cite{blauth_2023_improved} is not too complicated, it requires a good deal of technical care to sample this subgraph and analyze the expected penalty cost.
The tree construction in our algorithm is much simpler, which is also reflected in the analysis.

\bigskip

\sloppy
For \PCS, one can observe that a modification of the primal-dual approach by \textcite{goemans_1995_general} (like, e.g., in \cite{k_stroll}) yields a $2$-approximation.
\textcite{archer_2011_improved} showed that their techniques allow to beat the factor of $2$.
Concretely, using state-of-the-art \pathTSP algorithms as black-box subroutines, their algorithm is a $1.979$-approximation.
\textcite{an_path_tsp} observed that combining threshold rounding and the primal-dual approach by \textcite{goemans_1995_general} in the same manner as \textcite{goemans_2009_combining} did for \PCTSP, one can obtain an approximation guarantee of $\sfrac{1}{(1-\frac{1}{\beta'} e^{1-\sfrac{2}{\beta'}})}$, where~$\beta'$ denotes the approximation guarantee of an LP-relative approximation algorithm for \pathTSP that is used in a black-box way.
Using the best known bound on $\beta'$ of roughly $1.528$ \cite{traub_path_tsp, zhong_path_tsp} yields an approximation guarantee of approximately $1.926$.

\subsection{Related results}

Alongside the general version targeted here, \PCTSP{} was studied in special metric spaces.
A PTAS is known for graph metrics in planar graphs \cite{bateni_2011_planar} and in metrics with bounded doubling dimension \cite{chan_2020_unified}.
For asymmetric edge costs satisfying the triangle inequality, a $\lceil \log(|V|) \rceil$-approximation is known \cite{nguyen_2013_asymmetric}.

Besides \PCTSP{}, there is a wide class of other prize-collecting \TSP{} variants, most of which originate from the work of \textcite{balas_1989_prize}.
Although \PCTSP{} can be seen as the main variant in this problem class, there are other variants that include a lower bound on some minimum prize money that needs to be collected \cite{balas_1989_prize,ausiello_2000_salesmen,ausiello_2007_prize,garg_2005_saving, k_stroll}, or an upper bound on the distance that can be traveled \cite{blum_2003_approx,chekuri_2012_improved,paul_2020_budgeted,paul_2022_errandum,dezfuli_2022_combinatorial}.

Prize-collecting versions have also been studied for other classical combinatorial optimization problems.
The most prominent example is the prize-collecting Steiner tree problem, which admits a $1.79$-approximation \cite{ahmadiPCSteinerTree}, thereby going beyond the integrality gap of $2$ of the natural linear programming relaxation.
The situation is similar for the more general prize-collecting Steiner forest problem (see, e.g., \cite{hajiaghayi_2006_prize}), which admits a $2$-approximation~\cite{ahmadi20232approximation} while the integrality gap of the natural linear programming relaxation is known to be at least $\sfrac94$~\cite{koenemann_2017_intgap}.
Interestingly, the aforementioned lower bound of $\sfrac94$ is strictly larger than the gap of $2$ that the natural linear programming relaxation for the Steiner forest problem admits, indicating that prize-collecting aspects may in some cases make the problem strictly harder to approximate.
To date, no such separation is known for \TSP and \PCTSP, or \pathTSP and \PCS.

\subsection{Structure of the paper}

After describing our \PCTSP algorithm in more detail in \cref{sec:algorithm}, we first bound its approximation ratio by the golden ratio ${\sfrac{(1+\sqrt{5})}{2}} \approx 1.618$ through a simple analysis in \cref{sec:golden_ratio}.
In \cref{sec:below1.6}, we show that a minor adaption of this algorithm and a slightly more involved analysis allow us to push the approximation guarantee to $1.599$ as in \cref{thm:main}.
Finally, \cref{sec:pc_stroll} is devoted to adapting our algorithm to the setting of \PCS.
We first show in \cref{subsec:PCS_53} that our novel approach results in a $\sfrac53$-approximation, and later in \cref{subsec:PCS_combination} provide details on how to slightly beat this guarantee in order to obtain \cref{thm:PCS}. %
\section[Our algorithm for PCTSP]{Our algorithm for \PCTSP}\label{sec:algorithm}

Our algorithm for \PCTSP follows the basic idea of the Christofides-Serdyukov algorithm for \TSP, which is to combine a spanning tree $T$ with a shortest $\odd(T)$-join%
\footnote{
For a graph $G=(V,E)$, we denote by $\odd(G)\coloneqq \{v\in V\colon \deg_G(v) \text{ odd}\}$ the set of vertices with odd degree in $G$.
Furthermore, for $Q\subseteq V$ of even cardinality, a $Q$-join is a set of edges that has odd degree precisely at vertices in $Q$. (For metric cost functions there is always a minimum cost $Q$-join that is a matching.)
}%
, and shortcut an Eulerian tour in the resulting even-degree graph to obtain a cycle.
The operation of adding a shortest $\odd(T)$-join to a tree $T$ is also known as \emph{parity correction}, as it results in a graph in which every vertex has even degree.
Typically, for an even cardinality set $Q$ (in particular, for $Q=\odd(T)$), the cost of a shortest $Q$-join is bounded by the cost $c^\top z$ of a point $z$ that is feasible for the dominant of the $Q$-join polytope, which is given by (see \cite{edmonds_1973_matching})
\begin{equation}\label{eq:Qjoin}
\PQJ \coloneqq \{x\in\mathbb{R}^E\colon x(\delta(S))\geq 1\ \forall S\subseteq V \text{ with } |S\cap Q| \text{ odd}\}\enspace.
\end{equation}

We use this approach with two additional variations:
First, because our setting allows excluding some vertices in the returned tour, we use trees~$T$ that may not span all vertices in $V$.
Second, we follow what is known as a Best-of-Many approach, i.e., we construct a polynomial-size set of trees, construct a tour from each of them, and return the best.
For the analysis of such an approach, one typically provides a distribution over the involved trees and analyzes the expected cost of the returned tour---which implies the same bound on the best tour. We will do so here.

We base our tree construction on the following decomposition lemma, which we restate here in the form given in \cite{blauth_2023_improved}.
For completeness, in \cref{sec:tree_decomp_proof} we replicate the proof provided in \cite{blauth_2023_improved} with minor simplifications (as the proof in \cite{blauth_2023_improved} shows a generalized version of \cref{lem:tree_decomposition}).
As mentioned earlier, this proof very closely follows proofs for packing branchings in a directed multigraph (see \cite[Theorem 2.6]{bang-jensen_1995_preserving} and \cite[Theorem~3.1]{post_2015_lpbased}).

\begin{lemma}[{\cite[Lemma 12]{blauth_2023_improved}}]\label{lem:tree_decomposition}
Let $(x, y)$ be a feasible solution of the \ref{eq:rel_PCTSP}.
We can in polynomial time compute a set $\mathcal{T}$ of trees that all contain the root $r$, and weights $\mu \in [0,1]^\mathcal{T}$ such that $\sum_{T\in\mathcal{T}}\mu_T = 1$,
\begin{equation*}
\sum_{T\in \mathcal{T}}\mu_T \cdot \chi^{E[T]} \leq x\enspace,
\qquad\text{and}\qquad
\forall v\in V\colon
\sum_{T\in\mathcal{T}\colon v\in V[T]}\mu_T = y_v \enspace.%
\footnote{
For a set $F \subseteq E$ of edges, we denote by $\chi^F \in \{0,1\}^E$ the incidence vector of $F$.
}
\end{equation*}
\end{lemma}

As also mentioned in \cref{fn:simple_2_approx}, this lemma gives rise to an immediate $2$-approximation for \PCTSP that follows the framework described above:
Choosing a tree $T\in \mathcal{T}$ with probability $\mu_T$, and performing parity correction by doubling the tree gives a tour of expected length at most $2 c^{\transpose} x$, while the expected penalty incurred for vertices that are not visited can easily be seen to be $\pi^{\transpose} (1-y)$.

It is an intriguing question whether parity correction can be done at expected cost $\sfrac{c^{\transpose} x}{2}$.
Such a result would immediately lead to a $\sfrac32$-approximation algorithm for \PCTSP, matching the guarantee of the \citeauthor{christofides_1976_worst}-\citeauthor{serdyukov_1978_onekotorykh} algorithm for \TSP.
While in the setting of classical \TSP, we have $\sfrac{x}{2}\in\PQJ$ for any set $Q$ of even cardinality (and can thus bound the cost of parity correction for \emph{any} tree $T$ by $\sfrac{c^{\transpose} x}{2}$), this is no longer true in the prize-collecting setting.

Given a tree $T$, we first prune it to obtain a tree $T'$ (in a way we will shortly explain) and then perform parity correction. To analyze the cost of the parity correction, we construct a point $z\in\PQJ[\odd(T)]$ that is of the form
\begin{equation}\label{eq:combination}
z = \alpha\cdot x + \beta\cdot \chi^{E[T']}
\end{equation}
for some coefficients $\alpha,\beta\in\mathbb{R}_{\geq 0}$. (In fact, we will choose different coefficients $\beta$ for different parts of the tree.)
As the existence of cuts $S\subsetneq V$ for which both $x(\delta(S))$ and $\delta_{T'}(S)$ are small may require to choose $\alpha$ and $\beta$ large, we preprocess both the solution $(x,y)$ of the \ref{eq:rel_PCTSP} that we use as well as the trees that we obtain from it through \cref{lem:tree_decomposition}.

In our first preprocessing step, we get rid of cuts $S$ for which $x(\delta(S))$ is very small.
To this end, observe that our algorithm may always drop vertices $v\in V\setminus\{r\}$ for which $y_v$ is very small.
Concretely, if our tour does not visit a vertex $v$ with $y_v\leq \delta$ for some $\delta\in[0,1)$, we pay a penalty of $\pi_v$, which is at most a factor of $\sfrac1{(1-\delta)}$ larger than the fractional penalty $\pi_v(1-y_v)$ occurring in the LP objective.
Thus, if we aim for an $\alpha$-approximation algorithm, we may safely choose $\delta=1-\sfrac1\alpha$ and drop all vertices with $y_v\leq \delta$.
Crucially for our analysis, we can also perform this dropping on the level of solutions of the \ref{eq:rel_PCTSP} by using the so-called \emph{splitting off} technique \cite{lovasz_1976_connectivity,mader_1978_reduction,frank_1992_on}. For a fixed vertex $v \in V$, splitting off allows to decrease the $x$-weight on two well-chosen edges $\{v,s\}$ and $\{v,t\}$ incident to $v$ (and thereby also the value $y_v$) while increasing the weight on the edge $\{s,t\}$ by the same amount, without affecting feasibility for the \ref{eq:rel_PCTSP}.%
\footnote{
Generally, one can even guarantee that minimum $s$-$t$ cut sizes are preserved for all $s,t\in V\setminus \{v\}$.
Feasibility for the \ref{eq:rel_PCTSP} is already maintained by preserving minimum $r$-$u$ cut sizes for all $u\in V\setminus \{r, v\}$.
}
Note that such a \emph{feasible splitting} at $v$ does not increase the cost of the solution by the triangle inequality. A sequence of feasible splittings at $v$ that result in $y_v=0$ is called a \emph{complete splitting}.
Complete splittings always exist \cite{frank_1992_on} and can be found in polynomial time through  a polynomial number of minimum $s$-$t$ cut computations by trying all candidate pairs of edges (see, e.g., \cite{nagamochi_1997_complete,nagamochi_2006_fast} for more efficient procedures).
Summarizing the above directly gives the following.

\begin{theorem}[Splitting off]\label{thm:complete_splitting_pctsp}
Let $(x^*,y^*)$ be a feasible solution of the \ref{eq:rel_PCTSP}. Let $v \in V \setminus \{r\}$.
There is a deterministic algorithm that computes in polynomial time a complete splitting at $v$, i.e., a sequence of feasible splittings at $v$ resulting in a feasible solution $(x,y)$ of the \ref{eq:rel_PCTSP} with $y_v = 0$, as well as
$$
c^{\transpose} x \leq c^{\transpose} x^*
\qquad\text{and}\qquad
\forall u\in V \setminus \{v\} \colon\ y_u = y^*_u	\enspace.
$$
\end{theorem}
Our first preprocessing step then consists of repeatedly applying \cref{thm:complete_splitting_pctsp} to vertices $v \in V \setminus \{r\}$ with $y_v < \delta$ for some parameter $\delta \in [0,1)$ that we fix later.

Our second preprocessing step affects the trees that we obtain through \cref{lem:tree_decomposition}.
While the first preprocessing step guarantees that there are no non-trivial cuts $S$ with very small $x(\delta(S))$, we also want to eliminate cuts $S$ with moderately small $x(\delta(S))$ for which $T$ has only a single edge in $\delta(S)$, so that the combination $z$ defined in \eqref{eq:combination} gets a significant contribution from at least one of $x$ or $\chi^{E[T]}$ on every non-trivial cut.
To this end, we use a pruning step defined as follows (also see \cref{fig:core}).

\begin{definition}[Core]
For a fixed solution $(x,y)$ of the \ref{eq:rel_PCTSP}, a tree $T$ containing the root vertex, and a threshold $\gamma$, the \emph{core} of $T$ with respect to $\gamma$, denoted by $\core(T,\gamma)$, is the inclusion-wise minimal subtree of $T$ that spans all vertices $v\in V[T]$ with $y_v\geq \gamma$.
\end{definition}

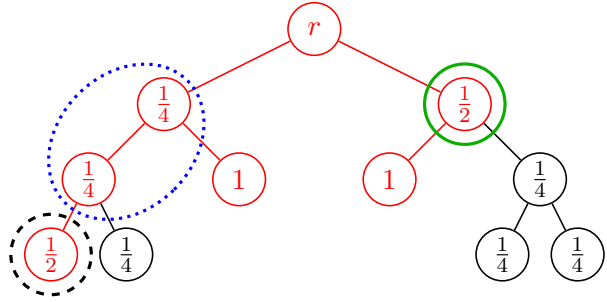
\begin{figure}[!t]
\centering
\begin{tikzpicture}[
    level 1/.style={sibling distance=40mm, level distance=10mm},
    level 2/.style={sibling distance=20mm, level distance=10mm},
    level 3/.style={sibling distance=10mm, level distance=10mm},
    every node/.style = {shape=circle, draw, align=center, minimum size=0.7cm, inner sep=0pt},
    every path/.style = {semithick}]

\node[red] {$r$}
    child [edge from parent/.style={draw=red}] {node[red] (n1) {$\frac{1}{4}$}
        child {node[red] (n2) {$\frac{1}{4}$}
            child {node[red] (n3) {$\frac12$}}
            child [edge from parent/.style={draw=black}] {node[black] {$\frac{1}{4}$}}
        }
        child  {node[red] {$1$}}
    }
    child [edge from parent/.style={draw=red}] {node[red] (n4) {$\frac{1}{4}$}
        child {node[red] {$1$}}
        child [edge from parent/.style={draw=black}] {node[black] {$\frac{1}{4}$}
            child [edge from parent/.style={draw=black}] {node[black] {$\frac{1}{4}$}}
            child [edge from parent/.style={draw=black}] {node[black] {$\frac{1}{4}$}}
        }
    };
\coordinate (center) at ($(n1)!.5!(n2)$);  %
\node [draw=blue, dotted, very thick, ellipse, minimum width=1.8cm, minimum height=2.3cm, rotate=-45, anchor=center] at (center) {};  %
\node [dashed, draw=black, very thick, fit=(n3), inner sep=0.5pt] {};
\node [draw=darkgreen!50!green, very thick, fit=(n4), inner sep=0.5pt] {};

\end{tikzpicture}
\caption{The core $T'=\core(T,\sfrac12)$ of the underlying tree $T$ at threshold $\gamma=\sfrac12$ is highlighted in red,  where $y_v$ is shown for each node $v \in V[T]$. We emphasize the different situations that can occur in terms of cuts: For the dotted blue cut $S$, $x(\delta(S))$ may be small as it does not contain a vertex with $y_v \ge \sfrac12$, however there are at least three edges of $T'$ in $\delta(S)$ that make up for this. The dashed black cut is a cut $S$ with only one tree edge in $\delta(S)$, however it contains a vertex with $y_v \ge \sfrac12$ so $x(\delta(S))$ is large. Finally, for the solid green cut, $x(\delta(S))$ may be small and there are only two tree edges in $\delta(S)$, however for this cut, $|\delta_{T'}(S)|$ is even, so there is no corresponding constraint in the dominant of the $\odd(T')$-join polytope.}
\label{fig:core}
\end{figure}

Indeed, if $T'=\core(T,\gamma)$, we know that for any non-empty cut $S\subsetneq V[T']$, we either have $x(\delta(S))\geq 2\gamma$, or $|\delta_T(S)|>1$.
Additionally, the only relevant such cuts $S$ in terms of parity correction on $T'$ are those with $|S\cap \odd(T')|$ odd, which is well-known to be equivalent to $|\delta_{T'}(S)|$ being odd because
\begin{align}
|S\cap \odd(T')|
\equiv \sum_{v\in S}\deg_{T'}(v)
= 2 \cdot \left|E[T']\cap\binom{S}{2}\right| + |\delta_{T'}(S)|
\equiv |\delta_{T'}(S)| \pmod{2}\enspace.\label{eq:odd}
\end{align}
Thus, for cuts $S$ with $|S\cap \odd(T')|$ odd, $|\delta_{T'}(S)|>1$ implies $|\delta_{T'}(S)|\geq3$, thereby further boosting the load on $\delta(S)$ in $z$ for this case (also see \cref{fig:core} for examples of the different types of cuts that may appear).
Altogether, we are now ready to state our new algorithm for \PCTSP, \cref{alg:improved_algo}.

\begin{algorithm2e}[!ht]
\caption{Our new approximation algorithm for \PCTSP\strut}\label{alg:improved_algo}
\SetKw{Input}{Input:}
\Input{\PCTSP{} instance $(G,r,c,\pi)$ on $G=(V,E)$, $\drop \in [0,1)$.}
\smallskip
\begin{stepsarabic}

\item\label{step:sol} Compute an optimal solution $(x^*, y^*)$ of the \ref{eq:rel_PCTSP}.

\item\label{algitem:split} Perform complete splittings at all $v \in V$ with $y_v < \delta$ (see \cref{thm:complete_splitting_pctsp}), resulting in a feasible solution $(x, y)$ of the \ref{eq:rel_PCTSP}.

\item\label{algitem:tree_decomp} Compute a set $\mathcal{T}$ of trees through \cref{lem:tree_decomposition} applied to $(x,y)$.

\item\label{algitem:trees} Let
$$
\mathcal{T}' = \bigcup_{\gamma \in \{y_v\colon v\in V\}} \{\core(T, \gamma)\colon T\in\mathcal{T}\}\enspace.
$$
\end{stepsarabic}
\Return{
Best tour found by doing parity correction on all trees in $\mathcal{T}'$.
}
\end{algorithm2e}

We remark that \cref{alg:improved_algo} can be implemented to run in polynomial time. Indeed, by \cref{lem:tree_decomposition} and \cref{thm:complete_splitting_pctsp}, \cref{algitem:split,algitem:tree_decomp,algitem:trees} run in polynomial time. Moreover, we can compute an optimum solution to the \ref{eq:rel_PCTSP} in polynomial time as the seperation problem can be solved by computing minimum $r$-$v$ cut sizes for each $v \in V \setminus \{r\}$.

We show in the next section that there is a constant $\delta$ for which \cref{alg:improved_algo} is a $\sfrac{\smash{(1+\sqrt{5})}}{2}$-approximation algorithm.
To go beyond that and prove \cref{thm:main}, we will later allow an instance-specific choice of $\delta$.

\section[A \texorpdfstring{$\sfrac{{(1+\sqrt{5})}}{2}$-}{golden ratio }approximation guarantee for PCTSP]{\boldmath A $\sfrac{\smash{(1+\sqrt{5})}}{2}$-approximation guarantee for \PCTSP}\label{sec:golden_ratio}

In this section, we prove the following result that gives the golden ratio as approximation guarantee for \PCTSP.
This is slightly weaker than what \cref{thm:main} claims, but the proof is simple and illustrates our main ideas.

\begin{theorem}\label{thm:weak}
\cref{alg:improved_algo} is an $\alpha$-approximation algorithm for \PCTSP with
\[
\alpha \coloneqq \max \left\{ \frac{5-2\delta}{3-\delta}, \frac{3-\delta}{2-\delta}, \frac{1}{1-\delta} \right\} \enspace.
\]
In particular, for $\delta=\sfrac{3-\sqrt5}{2}\approx0.382$, we get $\alpha=\sfrac{(1+\sqrt{5})}{2}\approx 1.618$.
\end{theorem}

Throughout this section, we fix a solution $(x,y)$ of the \ref{eq:rel_PCTSP} that was obtained from an optimal solution $(x^*, y^*)$ through complete splittings as in \cref{algitem:split} of \cref{alg:improved_algo}, and we fix a set $\mathcal{T}$ of trees with weights $(\mu_T)_{T\in\mathcal{T}}$ that is obtained in \cref{algitem:tree_decomp}, i.e., through \cref{lem:tree_decomposition} applied to $(x,y)$.
Moreover, we sample a random tree $T$ from the set $\mathcal{T}'$ constructed in \cref{algitem:trees} of \cref{alg:improved_algo} as follows:
For a fixed value $\kappa\in [\delta, 1]$, sample a threshold $\gamma\in [\delta, \kappa]$ such that for any $t\in [\delta, \kappa]$, we have
\begin{equation}\label{eq:cdf}
\Prob[\gamma \leq t] = \frac{3-\delta-\kappa}{3-\delta-t}\enspace.
\end{equation}
Independently, sample a tree $T_0\in\mathcal{T}$ with marginals given by $(\mu_T)_{T\in\mathcal{T}}$.
Then, define
\begin{equation}\label{eq:random_tree}
T \coloneqq \core(T_0,\gamma)\enspace.
\end{equation}
By definition of the core, it is clear that $T\in\mathcal{T}'$ even if $\gamma\not\in\{y_v\colon v\in V\}$.
To prove \cref{thm:weak}, we bound the expected cost of a tour constructed from $T$ by parity correction.
We remark that for proving \cref{thm:weak}, we only need $\kappa=1$; we nonetheless proceed in this generality here to be able to reuse some of the following statements in a proof of \cref{thm:main}.
To start with, we bound the expected tour length.

\begin{lemma}\label{lem:length}
Let $T=\core(T_0, \gamma)$ be a random tree generated as described in and above \eqref{eq:random_tree}, and let $C$ be the cycle obtained through parity correction on $T$ and shortcutting an Eulerian walk in the resulting graph.
Then
$$
\Exp[c(E[C])] \leq \frac{7-2\delta-2\kappa}{3-\delta}\cdot c^{\transpose} x^*\enspace.
$$
\end{lemma}

\begin{proof}
Let $\eta_1>\ldots>\eta_k$ such that $\{\eta_1,\ldots,\eta_k\}=\{y_v\colon v\in V\}$.
Define
$$
E_i \coloneqq \begin{cases}
E[\core(T_0,1)]& \text{for } i = 1\\
E[\core(T_0,\eta_i)]\setminus E[\core(T_0, \eta_{i-1})] & \text{for } i \in\{2,\ldots,k\}
\end{cases}\enspace.
$$
We refer to \cref{fig:Ei} for an illustration of the sets $E_i$.
\definecolor{aoe}{rgb}{0.0,0.42,0.24}%
\begin{figure}[tb]%
\centering %
\begin{tikzpicture}[
    level 1/.style={sibling distance=40mm, level distance=10mm},
    level 2/.style={sibling distance=20mm, level distance=10mm},
    level 3/.style={sibling distance=10mm, level distance=10mm},
    every node/.style = {shape=circle, draw, align=center, minimum size=0.7cm, inner sep=0pt},
    every path/.style = {semithick}]

\node[red] {$r$}
    child [edge from parent/.style={draw=red}] {node[red] {$\frac{1}{4}$}
        child [edge from parent/.style={draw=blue}] {node[blue]  {$\frac{1}{4}$}
            child {node[blue] {$\frac{1}{2}$}}
            child [edge from parent/.style={draw=aoe}] {node[aoe] {$\frac{1}{4}$}}
        }
        child  {node[red] {$1$}}
    }
    child [edge from parent/.style={draw=red}] {node[red] {$\frac{1}{4}$}
        child {node[red] {$1$}}
        child [edge from parent/.style={draw=aoe}] {node[aoe] {$\frac{1}{4}$}
            child [edge from parent/.style={draw=aoe}] {node[aoe] {$\frac{1}{4}$}}
            child [edge from parent/.style={draw=aoe}] {node[aoe] {$\frac{1}{4}$}}
        }
    };
\end{tikzpicture}%
\caption{%
Here, the variables $\eta_i$ are given by the sequence $1$, $\sfrac12$, $\sfrac14$.
Edges in $E_1$ are drawn in red, those in $E_2$ in blue, and those in $E_3$ in green.
We also color the nodes for intuition, however the sets $E_i$ consist solely of edges.
}%
\label{fig:Ei}
\end{figure}
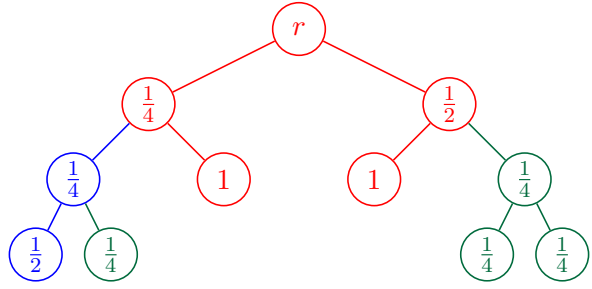%
In particular, this definition implies that
\begin{equation}
\label{eq:cost_T}
c(E[T]) = c\big(E[\core(T_0, \gamma)]\big) = \sum_{i\in[k]\colon \eta_i \geq \gamma} c(E_i) \enspace.
\end{equation}
To bound the cost of parity correction on $T$, we claim that
$$
z \coloneqq \frac1{3-\delta} \cdot x + \sum_{i \in [k] \colon \eta_i \geq \gamma} \left(1-\frac{2\eta_i}{3-\delta}\right) \cdot \chi^{E_{i}}
$$
lies in the dominant of the $\odd(T)$-join polytope.
This implies that a shortest $\odd(T)$-join $J$ has length
$$
c(E[J])) \leq \frac1{3-\delta} \cdot c^{\transpose} x + \sum_{i \in [k] \colon \eta_i \geq \gamma} \left(1-\frac{2\eta_i}{3-\delta}\right) \cdot c(E_{i})\enspace.
$$
Combining this with \eqref{eq:cost_T} and taking expectations immediately gives the desired bound:
\begin{align*}
\Exp[c(E[C])]
& \leq \frac1{3-\delta} \cdot c^{\transpose} x + \sum_{T_0\in\mathcal{T}} \mu_{T_0}\sum_{i=1}^k\Prob[\gamma\leq\eta_i]\cdot \left(2-\frac{2\eta_i}{3-\delta} \right) \cdot c(E_{i})\\
\tag{using \eqref{eq:cdf}}
& = \frac1{3-\delta} \cdot c^{\transpose} x + \sum_{T_0\in\mathcal{T}} \mu_{T_0}\sum_{i=1}^k\frac{6-2\delta-2\kappa}{3-\delta} \cdot c(E_{i})\\
& = \frac1{3-\delta} \cdot c^{\transpose} x + \frac{6-2\delta-2\kappa}{3-\delta}\cdot \sum_{T_0\in\mathcal{T}} \mu_{T_0}c(E[T_0])\\
\tag{using \cref{lem:tree_decomposition}}
& \leq \frac{7-2\delta-2\kappa}{3-\delta}\cdot c^{\transpose} x\enspace.
\end{align*}

By the construction of $(x,y)$ from $(x^*, y^*)$ through splitting off, we have $c^\top x\leq c^\top x^*$, hence the above implies the desired.
It thus remains to show that $z\in \PQJ[\odd(T)]$, i.e., that $z(\delta(S))\geq 1$ for every $S\subseteq V$ with $|S \cap \odd(T)|$ odd.
As remarked in \eqref{eq:odd}, $|S \cap \odd(T)|$ is odd if and only if $|\delta_{T}(S)|$ is.
If $|\delta_{T}(S)| \ge 3$,
\begin{align*}
z(\delta(S))
= \frac1{3-\delta} \cdot x(\delta(S)) + \sum_{i \in [k]\colon \eta_i \ge \gamma} \left(1-\frac{2\eta_i}{3-\delta}\right)\cdot |E_{i}\cap \delta_{T}(S)|
\ge \frac{2\delta}{3-\delta}+ \left(1-\frac{2}{3-\delta}\right)\cdot 3
= 1 \enspace,
\end{align*}
where we used that $x(\delta(S))\geq 2\delta$ because $y_v\geq \delta$ for all $v\in S$, and that $\eta_i\leq 1$ for each $i \in [k]$.
Otherwise, we have $|\delta_{T}(S)| = 1$.
Let $i \in [k]$ such that $\delta_T(S)\subseteq E_i$.
Then $x(\delta(S))\geq 2\eta_i$, hence also in this case
\[
z(\delta(S)) \ge \frac{2 \eta_i}{3-\delta} + \left(1-\frac{2\eta_i}{3-\delta}\right) = 1 \enspace.\qedhere
\]
\end{proof}

Next, we analyze the expected penalty incurred when starting with such a random tree $T$.

\begin{lemma}\label{lem:penalty}
	Let $T=\core(T_0, \gamma)$ be a random tree generated as described in and above \eqref{eq:random_tree}.
    Then, for every $v \in V$, we have
	\begin{equation*}
	\Prob[v \in V[T]] \geq \begin{cases}
		0 & \text{if } y^*_v \in [0,\delta) \\
		y^*_v \cdot \frac{3-\delta-\kappa}{3-\delta-y^*_v} & \text{if }  y^*_v \in [\delta, \kappa] \\
		y^*_v & \text{if } y^*_v \in (\kappa,1]
	\end{cases}\enspace.
	\end{equation*}
\end{lemma}

\begin{proof}
By construction, the solution $(x,y)$ has the property that for all $v\in V$, either $y_v=0$ or $y_v=y_v^*\geq \delta$.
Consequently, no tree $T_0$ in the family $\mathcal{T}$ generated through \cref{lem:tree_decomposition} contains vertices $v\in V$ with $y_v^* < \delta$, and thus the same holds for $T$.
Consequently, for such vertices $v\in V$, we get $\Prob[v \in V[T]] = 0$.
For vertices $v\in V$ with $y_v=y_v^*\geq \delta$, we have $\Prob[v\in V[T_0]]=y_v^*$ by \cref{lem:tree_decomposition}.
Hence,
\begin{equation*}
\Prob[v \in V[T]] \geq \Prob[v\in V[T_0]] \cdot \Prob[\gamma \le y_v^*] = y_v^* \cdot
\Prob[\gamma \le y_v^*] \enspace.
\end{equation*}
If $y_v^*>\kappa$, then $\Prob[\gamma \le y_v^*] = 1$ and $\Prob[v \in V[T]]=y_v^*$.
In the remaining case $y_v^*\in[\delta,\kappa]$, we use \eqref{eq:cdf} to obtain the desired.
\end{proof}

Together, \cref{lem:length,lem:penalty} allow us to conclude \cref{thm:weak}.

\begin{proof}[Proof of \cref{thm:weak}]
Let $T$ be a random tree generated as described in and above \eqref{eq:random_tree}, and let $C$ be the cycle obtained through parity correction on $T$ and shortcutting an Eulerian walk in the resulting graph.

Let $v \in V$. By \cref{lem:penalty}, if $y^*_v < \delta$, then $\Prob[v \notin V[C]] =1 \le \frac{1}{1-\delta} (1-y^*_v)$; if $y^*_v > \kappa$, then $\Prob[v \notin V[C]] \leq 1-y^*_v$.
If $\delta \le y^*_v \le \kappa$ then, again by \cref{lem:penalty},
\begin{align*}
	\Prob[v \notin V[C]]
	&\leq 1 - y_v^* \cdot \frac{3-\delta-\kappa}{3-\delta - y_v^*} \\
	&= \frac{(3-\delta)(1-y_v^*)-y_v^*(1-\kappa)}{3-\delta-y_v^*} \\
    \tag{because $y_v^*\leq\kappa\leq 1$}
	&\leq \frac{3-\delta}{3 - \delta - \kappa} \cdot (1- y_v^*) \enspace.
\end{align*}
Hence, together with \cref{lem:length} we get
\begin{align}
\nonumber
\Exp[c(E[C])  + \pi(V\setminus V[C])] & \leq
\frac{7-2\delta-2\kappa}{3-\delta}\cdot c^{\transpose} x^* + \max\left\{\frac{3-\delta}{3-\delta-\kappa}, \frac{1}{1-\delta}\right\}\cdot \pi^{\transpose} (1-y^*) \\
\label{eq:guarantee_kappa}&\leq
\max\left\{\frac{7-2\delta-2\kappa}{3-\delta}, \frac{3-\delta}{3-\delta-\kappa}, \frac{1}{1-\delta}\right\}\cdot \big(c^{\transpose} x^* + \pi^\transpose (1-y^*)\big) \enspace.
\end{align}
Independently of the realization of the involved random variables, the cycle $C$ is one that is generated in \cref{alg:improved_algo}.
The maximum in \eqref{eq:guarantee_kappa} is minimized for $\kappa = 1$ and $\delta=\sfrac{(3-\sqrt5)}{2}$, where it evaluates to $\sfrac{(1+\sqrt{5})}{2}$, thus giving the guarantee claimed in \cref{thm:weak}.
\end{proof} %

\section[Getting below \texorpdfstring{$1.6$}{1.6} for PCTSP]{\boldmath Getting below $1.6$ for \PCTSP}\label{sec:below1.6}

To improve upon the golden ratio approximation guarantee that we proved in \cref{sec:golden_ratio}, we exploit some remaining flexibility in the proof:
Sampling not only the threshold $\gamma$ but also $\delta$ from a distribution allows for trading off costs better than before.
The choice of distribution here is not best possible (though close to best possible, see \cref{rem:optimal_distr}), but designed to demonstrate that the true approximability of \PCTSP is below $1.6$ in a way that reduces the use of computers to the evaluation of a ``simple'' function that does not involve integrals.
To still obtain a deterministic algorithm, we also show how actually sampling $\delta$ can be avoided by trying polynomially many instance-dependent values.

\begin{proof}[Proof of \cref{thm:main}]
	For constants $\kappa_0$ and $\kappa$ to be fixed later, we sample $\delta\in[\kappa_0, \kappa]$ from a distribution with density
	$
	f(\delta) = \nu \cdot (3-\delta)(\kappa-\delta)^{2.2}
	$,
	where
	\begin{align*}
		\nu = \left(\int_{\kappa_0}^{\kappa} (3-\delta)(\kappa-\delta)^{2.2}\dd \delta \right)^{-1}
		= \left(\frac{(3-\kappa)(\kappa-\kappa_0)^{3.2}}{3.2} + \frac{(\kappa-\kappa_0)^{4.2}}{4.2}\right)^{-1}\enspace.
	\end{align*}
	Using this $\delta$, sample a tree $T$ as described in and above \eqref{eq:random_tree} (using the same $\kappa$ as here), and let $C$ be the cycle generated from parity correction on $T$ and shortcutting an Eulerian walk in the resulting graph.
	Using \cref{lem:length}, we get that the expected length of the cycle is
	\begin{align*}
		\Exp[c(E[C])]
		&\leq \nu \cdot \int_{\kappa_0}^\kappa (7-2\delta-2\kappa)(\kappa-\delta)^{2.2}\dd \delta \cdot c^\top x^* \\
		&= \underbrace{\nu\cdot\left(\frac{(7-4\kappa)(\kappa-\kappa_0)^{3.2}}{3.2} + \frac{2(\kappa-\kappa_0)^{4.2}}{4.2}\right)}_{\eqqcolon g(\kappa, \kappa_0)} {}\cdot c^\top x^*\enspace.
	\end{align*}
	Next, we bound the expected penalty. Let $v \in V$. By \cref{lem:penalty}, if $y^*_v < \kappa_0$, then $\Prob [v \notin V[C]] = 1 \le \frac{1}{1-\kappa_0}\cdot\left(1-y_v^*\right)$; if $y^*_v > \kappa$, then $\Prob [v \notin V[C]] \le 1 - y^*_v$. For $y_v^*\in [\kappa_0, \kappa]$, we again use \cref{lem:penalty} and obtain
	\begin{align*}
		\Prob[v \notin V[C]]
		&\le 1- y_v^* \cdot \nu \cdot \int_{\kappa_0}^{y_v^*} \frac{3-\delta-\kappa}{3-\delta-y_v^*} (3-\delta)(\kappa-\delta)^{2.2}\dd \delta\enspace.
	\end{align*}
	Now observe that for $y_v^* \leq \kappa$ the function $\delta\mapsto \phi(\delta) \coloneqq \frac{(3-\delta-\kappa)(3-\delta)}{3-\delta-y_v^*} = 3-\delta - \kappa + y_v^* - \frac{y_v^* (\kappa - y_v^*)}{3-\delta - y_v^*}$ is concave on $[\kappa_0, \kappa]$, hence for each $\delta \in [\kappa_0, \kappa]$, we have
	$$
	\phi(\delta) \geq
	\phi(\kappa_0)  \cdot \frac{\kappa-\delta}{\kappa-\kappa_0} +
	\phi(\kappa) \cdot \frac{\delta-\kappa_0}{\kappa-\kappa_0}\enspace.
	$$
	Plugging in this bound and evaluating the involved integrals gives
	\begin{align*}
		\Prob[v\notin V[C]] &\leq
		1 - \frac{y_v^* \cdot \nu}{\kappa-\kappa_0} \cdot \Bigg(
		\phi(\kappa_0)  \cdot \int_{\kappa_0}^{y_v^*}(\kappa-\delta)^{3.2}\dd\delta
		+
			\phi(\kappa) \cdot \int_{\kappa_0}^{y_v^*}(\delta-\kappa_0)(\kappa-\delta)^{2.2}\dd\delta
			\Bigg)\\
		&=
		1 - \frac{y_v^* \cdot \nu}{\kappa-\kappa_0} \cdot \Bigg(
		\left(\phi(\kappa_0) - \phi(\kappa) \right)   \cdot \frac{(\kappa-\kappa_0)^{4.2}-(\kappa-y_v^*)^{4.2}}{4.2} \\
		&\qquad\qquad\qquad\qquad\qquad{} +
			\phi(\kappa) \cdot (\kappa - \kappa_0) \cdot
			\frac{(\kappa-\kappa_0)^{3.2}-(\kappa-y_v^*)^{3.2}}{3.2}
			\Bigg) 
		\eqqcolon h_{y_v^*}(\kappa, \kappa_0)\enspace.
	\end{align*}
	Thus
	\begin{align}\label{eq:final_term}
		\Exp[\pi(V\setminus V[C])] \leq \underbrace{\max\left\{\frac{1}{1-\kappa_0}, \max_{y\in [\kappa_0, \kappa]} \frac{h_y(\kappa, \kappa_0)}{1-y}\right\}}_{\eqqcolon h(\kappa, \kappa_0)}{}\cdot \pi^\top (1-y^*)\enspace,
	\end{align}
	and we get a bound on the expected total cost of the form
	$$
	\Exp[c(E[C]) + \pi(V\setminus V[C])] \leq \max\{g(\kappa, \kappa_0), h(\kappa, \kappa_0)\} \cdot \big(c^\top x^* + \pi^\top(1-y^*)\big)\enspace.
	$$
	The latter maximum evaluates to slightly below $1.599$ for $\kappa_0 = 0.3724$ and $\kappa=0.9971$, thus giving the desired guarantee.
	To calculate the maximum in \eqref{eq:final_term}, we use that the derivative of \smash[b]{$y\mapsto \frac{h_y(\kappa, \kappa_0)}{1-y}$} on the interval $[\kappa_0, \kappa]$ can be bounded by a constant.
	Hence, $\max_{y\in [\kappa_0, \kappa]} \frac{h_y(\kappa, \kappa_0)}{1-y}$ can be approximated up to a minor error by evaluating the function for each $y$ in a sufficiently fine discretization of the interval $[\kappa_0, \kappa]$.
	
	Finally, note that the choice of $\delta$ can be derandomized by only trying the instance-specific values in the set $\{y_v ^*\colon v \in V\}$ obtained from the optimal LP solution $(x^*, y^*)$ that is used.
	Indeed, if $\delta \notin\{y_v^*\colon v \in V\}$, the bound on the expected cost of the cycle given in \cref{lem:length} improves (as long as $\kappa \ge \sfrac{1}{2}$) by using the minimal $\delta'$ in $\{y_v^*\colon v \in V, y_v^*\geq \delta\}$, whereas the bound on the expected penalty cost does not change.
\end{proof}

\begin{remark}\label{rem:optimal_distr}
Computational experiments based on discretizing a distribution over pairs $(\delta, \kappa)$ suggest that an analysis following the one in the above proof cannot achieve an approximation ratio of $1.59$.
We emphasize that this does not exclude that the actual approximation guarantee of \cref{alg:improved_algo} is below~$1.59$.
\end{remark} %
\section[Extending our approach to PCS]{Extending our approach to \PCS}\label{sec:pc_stroll}

\subsection[A straightforward \texorpdfstring{$\sfrac53$}{5/3}-approximation]{\boldmath A straightforward $\sfrac53$-approximation}
\label{subsec:PCS_53}

In this section we use the following linear programming relaxation for \PCS that can be obtained by starting with the \ref{eq:rel_PCTSP} and adjusting it to account for the fact that we want to find a stroll instead of a tour:%
\footnote{%
We use $\pi_s=\pi_t=0$ for convenience.
}
\begin{equation}\tag{\PCS LP relaxation}\label{eq:rel_PCS}
\begin{aligned}
\begin{array}{rrcll}
\min & \displaystyle \sum_{e\in E}c_ex_e + \sum_{v\in V}\pi_v (1-y_v)\\
     & x(\delta(v)) & = & 2y_v & \forall v\in V \setminus \{s,t\} \\
     & x(\delta(s)) & = & 1 \\
     & x(\delta(t)) & = & 1 \\
     & x(\delta(S)) & \geq & 2y_v & \forall S\subseteq V\setminus\{s, t\}, v\in S\\
     & x(\delta(S)) & \geq & 1 & \forall S\subseteq V\setminus\{t\}, s\in S\\
     & y_s & = & 1 \\
     & y_t & = & 1 \\
     & x_e & \geq & 0 & \forall e\in E\\
     & y_v & \geq & 0 & \forall v\in V\enspace.
\end{array}
\end{aligned}
\end{equation}
In particular, observe that we require a (fractional) degree of $1$ at $s$ and $t$, reflecting the fact that $s$ and $t$ are the fixed endpoints of the stroll that we are looking for.
Also note that we define $y_s=y_t=1$ to support our interpretation of the $y$-values as the extent to which a vertex is visited by the fractional solution.
Furthermore, in distinction from the \ref{eq:rel_PCTSP}, we enforce at least one unit of $x$-weight on $s$-$t$ cuts, as all of these are crossed at least once by every $s$-$t$ stroll.

As in our algorithm for \PCTSP, we also decompose a solution to the \ref{eq:rel_PCS} into trees.
To this end, we use the following analogue to \cref{lem:tree_decomposition}. 

\begin{lemma}\label{lem:tree_decomposition_PCS}
Let $(x, y)$ be a feasible solution of the \ref{eq:rel_PCS}.
We can in polynomial time compute a set $\mathcal{T}$ of trees that all contain the vertices $s$ and $t$, and weights $\mu \in [0,1]^\mathcal{T}$ such that $\sum_{T\in\mathcal{T}}\mu_T = 1$,
\begin{equation*}
\sum_{T\in \mathcal{T}}\mu_T \cdot \chi^{E[T]} = x\enspace,
\qquad\text{and}\qquad
\forall v\in V\colon
\sum_{T\in\mathcal{T}\colon v\in V[T]}\mu_T = y_v \enspace.
\end{equation*}
\end{lemma}

We remark that a polyhedral description of $s$-$t$ strolls very similar to the above \ref{eq:rel_PCS} along with the decomposition stated in \cref{lem:tree_decomposition_PCS} was recently first used to obtain improved approximation guarantees for another \TSP variation termed \emph{Ordered TSP}~\cite{armbruster_2024_ordered}.
For completeness, we provide a full proof of \cref{lem:tree_decomposition_PCS} in \cref{sec:tree_decomp_proof}. 
The lemma naturally leads to the following adaption of \cref{alg:improved_algo} to~\PCS.

\begin{algorithm2e}[!ht]
\caption{Adapting \cref{alg:improved_algo} to \PCS\strut}\label{alg:PCS}
\SetKw{Input}{Input:}
\Input{\PCS{} instance $(G,s,t,c,\pi)$ on $G=(V,E)$.}
\smallskip
\begin{stepsarabic}

\item\label{step:sol_PCS} Compute an optimal solution $(x, y)$ of the \ref{eq:rel_PCS}.

\item\label{algitem:tree_decomp_PCS} Compute a set $\mathcal{T}$ of trees through \cref{lem:tree_decomposition_PCS} applied to $(x,y)$.

\item\label{algitem:trees_PCS} Let
$$
\mathcal{T}' = \bigcup_{\gamma \in \{y_v\colon v\in V\}} \{\core(T, \gamma)\colon T\in\mathcal{T}\}\enspace.
$$
\end{stepsarabic}
\Return{
Best $s$-$t$ stroll found after parity correction on all trees in $\mathcal{T}'$.
}
\end{algorithm2e}

In the below, we show that \cref{alg:PCS} is a $\sfrac53$-approximation.
At first sight, taking intuition from the classical analyses of Christofides' algorithm for \TSP and our algorithm for \PCTSP, one may expect the approximation ratio of \cref{alg:PCS} for \PCS to be significantly worse when compared to the approximation ratio of \cref{alg:improved_algo} for \PCTSP, as the guarantee of the LP on cuts separating $s$ and $t$ is much weaker.
Thus, our guarantee of $\sfrac53$ may come as a surprise.
The reason we can achieve this can be explained as follows:
Let $z$ be the vector we wish to be in $\PQJ$ for $Q=\odd(T)\triangle\{s,t\}$ when bounding the cost of parity correction for a random tree $T$ from $\mathcal{T}'$.
Then, as in the analysis of \cref{alg:improved_algo} for \PCTSP, we will put additional value on $z_e$ for all $e \in T$.
But the additional challenge for \PCS is cuts separating~$s$ and~$t$ with \textit{even} parity.
This means any such cut gets additional value from \textit{two} edges of $T$, instead of one, the usual worrisome contribution of $T$ to a cut.
This additional boost is exactly what allows us to obtain the approximation ratio~$\sfrac53$.
More precisely, we obtain the following result.

\begin{theorem}\label{thm:PCS_unbalanced}
\cref{alg:PCS} returns an $s$-$t$ stroll $S=(V_S, E_S)$ satisfying
$$
c(E_S) + \pi(V\setminus V_S) \leq \frac53\cdot c^\top x + \frac32\cdot \pi^\top (1-y)\enspace.
$$
\end{theorem}

\begin{proof}
Similar to \cref{thm:weak}, we will analyze selecting a random tree $T$ from the set $\mathcal{T}$ generated in \cref{alg:PCS} precisely as defined in and above \eqref{eq:random_tree}, but for $\delta=0$ and $\kappa=0$.
This time, we perform parity correction to obtain odd degrees at $s$ and $t$, and even degrees elsewhere (i.e., parity correction at vertices in $Q\coloneqq \odd(T)\triangle \{s,t\}$ through a $Q$-join).
We bound the cost of parity correction as in the proof of \cref{lem:length}.
To demonstrate this, we show that the point
$$
z \coloneqq \frac1{3} \cdot x + \sum_{i \in [k] \colon \eta_i \geq \gamma} \left(1-\frac{2\eta_i}{3}\right) \cdot \chi^{E_{i}}
$$
(the same $z$ as in the aforementioned earlier proof, just for $\delta=0$) is feasible for $\PQJ$.
So, we need to show that $z(\delta(S))\geq 1$ for all cuts $S\subsetneq V$ where $|S\cap Q|$ is odd.

For non-empty such cuts $S$ with $S\subseteq V\setminus\{s,t\}$ (and, symmetrically, for those with $\{s,t\}\in S$), the analysis is exactly the same as in the previous proof.
For $s$-$t$ cuts $S$ (and again, by symmetry, for $t$-$s$ cuts $S$), we observe that
$$
|\delta(S)\cap T| \equiv |S\cap\odd(T)| \equiv |S\cap Q| + 1 \pmod{2} \enspace,
$$
hence for any cut with $|S\cap Q|$ odd, there are at least two tree edges in $\delta(S)$.
Therefore,
\begin{align*}
z(\delta(S)) &= \frac13 x(\delta(S)) + \sum_{i \in [k] \colon \eta_i \geq \gamma}\left(1-\frac{2\eta_i}{3}\right) |E_i\cap \delta(S)|
\\
&\geq \frac13 + \sum_{i \in [k] \colon \eta_i \geq \gamma}\frac{1}{3} |E_i\cap \delta(S)|
\tag{using $x(\delta(S))\geq 1$ and $\eta_i\leq 1$}
\\
&\geq 1\enspace.
\tag{using $\sum_{i \in [k] \colon \eta_i \geq \gamma}|E_i\cap \delta(S)| = |T\cap \delta(S)| \geq 2$}
\end{align*}
Consequently, following the proof of \cref{lem:length}, we can obtain an $s$-$t$ stroll $C$ satisfying
$$
\Exp[c([C])] \leq \frac53 \cdot c^\top x\enspace.
$$
Vertices are included in $C$ with the probabilities given by \cref{lem:penalty} for $\delta=0$ and $\kappa=1$, i.e.,
$$
\Prob[v\in V[C]] \geq y_v \cdot \frac{2}{3-y_v}\enspace,
$$
hence the expected penalty we pay can be bounded as follows:
\begin{align*}
\Exp[\pi(V\setminus V[C])]
&\leq \sum_{v\in V} \pi_v \left(1-y_v \cdot \frac{2}{3-y_v}\right)
= \sum_{v\in V} \frac{3\pi_v (1-y_v)}{3-y_v}
\leq \frac32 \cdot \pi^\top (1-y)\enspace.
\end{align*}
Thus, we indeed lose a factor of $\sfrac53$ on the LP edge cost, and a factor $\sfrac32$ on the LP penalty cost, giving us the claimed guarantee.
\end{proof}

\subsection{Combination with threshold rounding}
\label{subsec:PCS_combination}

As stated in \cref{thm:PCS_unbalanced}, our analysis of \cref{alg:PCS} for \PCS exhibits unbalanced losses on the edge cost and penalty terms.
We exploit that in the classical threshold rounding algorithm of \textcite{bienstock_1993_note}, originally introduced for \PCTSP, the threshold parameter can be chosen to give an imbalance in the other direction.
For \PCS, the straightforward adaption of the threshold rounding approach with threshold $\gamma$ is to consider an optimal solution $(x,y)$ of the \ref{eq:rel_PCS}, and return an $s$-$t$ path on the vertex set $\{v\in V\colon y_v\geq\gamma\}$.
\textcite{an_path_tsp} observed that one can obtain the following guarantee.

\begin{proposition}[{\cite[Section~4.1]{an_path_tsp}}]\label{prop:thresholding_PCS}
Let $(x,y)$ be a solution of the \ref{eq:rel_PCS}, and let $\beta'$ denote the LP-relative approximation guarantee of a polynomial-time algorithm for \pathTSP that can be accessed in a black-box way.
For every $\gamma\in(0,1)$, threshold rounding with threshold $\gamma$ results in a \PCS solution of cost at most
\[
\frac{\beta'}{\gamma} \cdot c^\top x + \frac{1}{1-\gamma} \cdot \pi^\top (1-y) \enspace.
\]
\end{proposition}

We recall that the currently best known LP-relative approximation algorithm for \pathTSP has an approximation factor of $\beta'<1.528$ \cite{traub_path_tsp, zhong_path_tsp}, giving enough room for a beneficial combination of the guarantees in \cref{thm:PCS_unbalanced,prop:thresholding_PCS} as follows.

\begin{proof}[Proof of \cref{thm:PCS}]
To start with, we compute an optimal solution $(x,y)$ of the \ref{eq:rel_PCS}.
This can be done in polynomial time, as the separation problem for the \ref{eq:rel_PCS} can be reduced to computing a minimum $s$-$t$ cut, and minimum $\{s,t\}$-$v$ cuts for all $v\in V\setminus \{s,t\}$.
Starting from this fixed LP solution $(x,y)$, we run \cref{alg:PCS} with probability $p=0.9922$, else we do threshold rounding with threshold $\gamma = 0.9561$.
Using \cref{thm:PCS} and \cref{prop:thresholding_PCS}, we get that the resulting solution has expected cost that can be bounded by 
\begin{multline*}
\left(p \cdot \frac53 + (1-p) \cdot \frac{1.528}{\gamma} \right) \cdot c^\top x + \left( p \cdot \frac32 + (1-p) \cdot \frac{1}{1-\gamma} \right) \cdot \pi^\top (1-y) \\
\le 1.6662 \cdot \left( c^\top x + \pi^\top (1-y) \right) \enspace.
\end{multline*}
Clearly, the better of the two computed solutions then also satisfies the above bound.
\end{proof}
\section[Proofs of Lemmas \ref{lem:tree_decomposition} and \ref{lem:tree_decomposition_PCS}]{Proofs of \cref{lem:tree_decomposition,lem:tree_decomposition_PCS}}\label{sec:tree_decomp_proof}

For the sake of this paper being self-contained, we provide proofs of \cref{lem:tree_decomposition,lem:tree_decomposition_PCS}. We do very closely follow the proof in \cite{blauth_2023_improved}, but incorporate some minor simplifications stemming from the fact that \cite{blauth_2023_improved} show a generalized version of the statements needed here.
More precisely, we show that \cref{lem:tree_decomposition} and \cref{lem:tree_decomposition_PCS} are a direct consequence of the following.

\begin{lemma}\label{lem:tree_decomposition_e0}
	Let $(x, y)$ be a feasible solution of the \ref{eq:rel_PCTSP}. Assume that there is an edge $e_0=\{r,v\} \in E$ with $x_{e_0} \ge 1$ and $y_v=1$. 
	We can in polynomial time compute a set $\mathcal{T}$ of trees that all contain the root $r$, and weights $\mu \in [0,1]^\mathcal{T}$ such that $\sum_{T\in\mathcal{T}}\mu_T = 1$,
	\begin{equation*}
		\sum_{T\in \mathcal{T}}\mu_T \cdot \chi^{E[T]} = x - \chi^{\{e_0\}}\enspace,
		\qquad\text{and}\qquad
		\forall v\in V\colon
		\sum_{T\in\mathcal{T}\colon v\in V[T]}\mu_T = y_v \enspace.
	\end{equation*}
\end{lemma}

Indeed, \cref{lem:tree_decomposition_PCS} is a direct consequence of \cref{lem:tree_decomposition_e0}. Starting with a solution of the \ref{eq:rel_PCS}, a solution of the \ref{eq:rel_PCTSP} can be obtained by increasing $x_{\{s,t\}}$ by $1$. Applying \cref{lem:tree_decomposition_e0} to this LP solution with $e_0 = \{s,t\}$ yields the desired tree decomposition. 

Second, to show that \cref{lem:tree_decomposition_e0} also implies \cref{lem:tree_decomposition} we can leverage a trick used, e.g., in \cite{karlin_2021_slightly, blauth_2023_improved} that we shortly discuss for completeness. Given a \PCTSP instance, construct an auxiliary instance by dividing the root vertex $r$ and all incident edges into two (with penalty zero and equal edge costs, respectively), and adding an edge $e_0$ of length $c(e_0)=0$ between the root and its copy.
Note that an LP solution for the original instance can be transformed to an LP solution of the same value for the thereby obtained auxiliary instance (to that we then apply \cref{lem:tree_decomposition_e0}) by setting $y_{v}=1$ for the copy of the root, $x_{e_0}=2-\frac{x(\delta(r))}{2}$, and distributing the $x$-weight on every edge incident to the root equally to its two copies in the auxiliary graph.
(Note that $x_{e_0} > 1$ if $x(\delta(r)) < 2$.)
The obtained set of trees in the auxiliary instance can be transformed to a set of trees in the original instance with the desired properties by contracting $e_0$ and possibly deleting edges of thereby obtained cycles.

It remains to prove \cref{lem:tree_decomposition_e0}.

\begin{proof}[Proof of \cref{lem:tree_decomposition_e0}]
We prove the statement by induction on $|V|$.
If $V = \{r,v\}$, then a feasible decomposition is given by one tree $T$ spanning $V$ with weight $\mu_T=1$.
If $|V|>2$, consider a vertex $z\in V\setminus e_0$ that minimizes $y_z$.
By \cref{thm:complete_splitting_pctsp}, we can efficiently obtain a complete splitting at $z$, i.e., a sequence of splitting operations that result in a solution $(x', y|_{V'})$ of the \ref{eq:rel_PCTSP} over $V'\coloneqq V\setminus \{z\}$.
Consequently, by the inductive assumption, we can in polynomial time compute a set $\mathcal{T}$ of trees and values $\mu_T$ with the desired properties; in particular, for every $v\in V'$,
$$
\sum_{T\in \mathcal{T}\colon v\in V[T]} \mu_T = y_v \enspace
$$
and $\sum_{T\in \mathcal{T}}\mu_T \cdot \chi^{E[T]} = x' - \chi^{\{e_0\}}$. We now undo the splitting operations at $z$ one after another and modify the trees in $\mathcal{T}$ accordingly.
Before we start undoing the splitting operations, we initialize auxiliary variables $\spare{v}=0$ for each $v \in V'$.

Consider a splitting operation on edges $e=\{z,u\}$ and $f=\{z,w\}$ with $e \neq f$ that reduces the weight on each of these edges by $\delta$.
Let $\mathcal{T}'\coloneqq \{T\in \mathcal{T}\colon \{u,w\}\in E[T]\}$.
Note that by $\sum_{T\in \mathcal{T}}\mu_T \cdot \chi^{E[T]} = x' - \chi^{\{e_0\}}$, we have $\sum_{T\in \mathcal{T}'}\mu_T \ge \delta$. (Clearly, this holds if $\{u,w\} \neq e_0$. If $\{u,w\} = e_0$, then $x_{e_0} \ge 1$ before the splitting operation, hence $x_{e_0} \ge 1+\delta$ after the splitting operation.)
If $\sum_{T\in \mathcal{T}'}\mu_T > \delta$, remove trees from $\mathcal{T}'$ until $\sum_{T\in \mathcal{T}'}\mu_T = \delta$ (this may require creating a copy of some tree and splitting its weight).
Now for each $T \in \mathcal{T}'$ do the following:
\begin{enumerate}
\item\label{item:not_contained} If $z \notin V[T]$, remove $\{u,w\}$ from $E[T]$ and add $\{z,u\}$ and $\{z,w\}$ to $E[T]$.

\item\label{item:contained} If $z \in V[T]$, remove $\{u,w\}$ from $E[T]$ and add either $\{z,u\}$ or $\{z,w\}$ to $E[T]$ such that $T$ remains acyclic.
If $\{z,u\}$ is added to $E[T]$ increase $\spare{w}$ by $\mu_T$.
Otherwise, increase $\spare{u}$ by $\mu_T$.
\end{enumerate}
Note that in case \ref{item:not_contained}, the total weight of trees containing $z$ increases by $\mu_T$.
Otherwise, either $\spare{u}$ or $\spare{w}$ increases by $\mu_T$.
Consequently, through the above operations (including the potential initial increase of $\spare{u}$), the sum $\sum_{T \in \mathcal{T}\colon z \in V[T]} \mu_T + \sum_{v \in V'} \spare{v}$ increases by $\delta$.
Since the splitting operation decreased the degree of $z$ by $2\delta$, we get, after reverting all splitting operations at $z$,
\[
\sum_{T \in \mathcal{T}\colon z \in V[T]} \mu_T + \sum_{v \in V'} \spare{v} = \frac{x(\delta(z))}{2} = y_z\enspace.
\]
Now, for each $w \in V'$ do the following:
If $\spare{w} > 0$, let $\mathcal{T}'' = \{T\in \mathcal{T}\colon w \in V[T], z \notin V[T]\}$. Note that
\[
\sum_{T\in \mathcal{T}\colon w\in V[T]} \mu_T = y_w \ge y_z = \sum_{T \in \mathcal{T}\colon z \in V[T]} \mu_T + \sum_{v \in V'} \spare{v} \enspace,
\]
where the inequality above holds because $z$ minimizes $y_v$ over all $v\in V\setminus\{r\}$, and the first equality holds by the inductive assumption and the fact that for every vertex in $V'$ the total weight of trees covering this vertex is unchanged under splittings and their reversal.

In particular, the above implies $\sum_{T\in\mathcal{T}''} \mu_T\geq \spare{w}$.
If this inequality is strict, remove trees from $\mathcal{T}''$ until $\sum_{T\in\mathcal{T}''} \mu_T = \spare{w}$ (again, this may require creating a copy of some tree and splitting its weight).
Now, for each $T \in \mathcal{T}''$, add $\{z,w\}$ to $E[T]$.
Note that this increases the total weight of the trees containing $z$ by $\spare{w}$.
Hence, after using up all spares, we get
\begin{equation}\label{eq:incident_trees}
\sum_{T \in \mathcal{T}\colon z \in V[T]} \mu_T = y_z\enspace.
\end{equation}

Note that throughout the above operations, the graphs $T \in \mathcal{T}$ are trees, and after reverting all splitting operations, $\sum_{T\in \mathcal{T}}\mu_T\chi^{E[T]} = x - \chi^{\{e_0\}}$ by construction.
Moreover, as mentioned above, for every vertex in $V'$, the total weight of trees covering this vertex is unchanged.
Hence, $\sum_{T\in\mathcal{T}\colon v\in V[T]}\mu_T = y_v$ is still satisfied for each $v \in V'$, and also for $z$ by~\eqref{eq:incident_trees}.

Finally, we note that our construction can be executed in polynomial time.
Indeed, by \cref{thm:complete_splitting_pctsp}, in each step of our inductive procedure, we have to revert less than $\operatorname{poly}(|V|)$ many splitting operations, which increases the total number of trees in $\mathcal{T}$ by an additive $\operatorname{poly}(|V|)$.
This implies that the size of $\mathcal{T}$ remains polynomially bounded throughout.
\end{proof} 
\begingroup
\setlength{\emergencystretch}{1em}
\printbibliography
\endgroup

\end{document}